\documentclass[12pt]{article} 
\usepackage{a4}
\usepackage{amsthm}
\usepackage{amsfonts}
\usepackage{algorithm2e}
\usepackage{cite}
\usepackage{epsfig}
\newcommand{\AAA}{$\mathcal{A}$}
\newcommand{\CC}{\mathcal{C}}
\newcommand{\GG}{\mathcal{G}}
\newcommand{\HH}{\mathcal{H}}
\newcommand{\NN}{{\mathbb N}}
\newcommand{\SSS}{\mathcal{S}}
\newtheorem{theorem}{Theorem}
\newtheorem{proposition}[theorem]{Proposition}
\newtheorem{lemma}[theorem]{Lemma}
\newtheorem{corollary}[theorem]{Corollary}
\newtheorem{conjecture}{Conjecture}
\newtheorem{problem}{Problem}

\begin{document}
\title{Recovering sparse graphs}
\author{Jakub Gajarsk\'y\thanks{Technical University Berlin. E-mail: {\tt jakub.gajarsky@tu-berlin.de}. This author's research was supported by the European Research Council (ERC) under the European Union's Horizon 2020 research and innovation programme (ERC Consolidator Grant DISTRUCT, grant agreement No 648527).}\textsuperscript{\hskip 1ex,}\footnotemark[3]\and
        Daniel Kr\'al'\thanks{Mathematics Institute, DIMAP and Department of Computer Science, University of Warwick, Coventry CV4 7AL, UK. E-mail: {\tt d.kral@warwick.ac.uk}. The work of this author was supported by the European Research Council (ERC) under the European Union’s Horizon 2020 research and innovation programme (ERC Consolidator Grant LADIST, grant agreement No 648509).}\textsuperscript{\hskip 1ex,}\thanks{This publication reflects only its authors' view; the European Research Council Executive Agency is not responsible for any use that may be made of the information it contains.}}
\date{}
\maketitle

\begin{abstract}
We construct a fixed parameter algorithm parameterized by $d$ and $k$ that 
takes as an input a graph $G'$ obtained
from a $d$-degenerate graph $G$ by complementing on at most $k$ arbitrary subsets of the vertex set of $G$ and
outputs a graph $H$ such that $G$ and $H$ agree on all but $f(d,k)$ vertices.

Our work is motivated by the first order model checking in graph classes that
are first order interpretable in classes of sparse graphs.
We derive as a corollary that if $\GG$ is a graph class with bounded expansion,
then the first order model checking is fixed parameter tractable in the class of all graphs that
can obtained from a graph $G\in\GG$ by complementing on at most $k$ arbitrary subsets of the vertex set of $G$;
this implies an earlier result that the first order model checking is
fixed parameter tractable in graph classes interpretable in classes of graphs with bounded maximum degree.
\end{abstract}

\section{Introduction}

The work presented in this paper is motivated by the line of research on algorithmic metatheorems,
general algorithmic results that guarantee the existence of efficient algorithms for wide classes of problems.
The most classical example of such a result is the celebrated theorem of Courcelle~\cite{bib-courcelle90} asserting
that every monadic second order property can be model checked in linear time in every class of graphs with bounded tree-width;
further results of this kind can be found in the survey~\cite{bib-kreutzer09}.
Specifically, our motivation comes from the first order model checking in sparse graph classes and
attempts to extend these results to classes of dense graphs with structural properties close to sparse graph classes.

The two very classical algorithms for the first order model checking in sparse graph classes
are the linear time algorithm of Seese~\cite{bib-seese96} for graphs with bounded maximum degree and
the linear time algorithm of Frick and Grohe~\cite{bib-frick01+} for planar graphs,
which can also be adapted to an almost linear time algorithm for graphs with locally bounded tree-width.
These results were extended to many other classes of sparse graphs,
in particular to graphs locally excluding a minor by Dawar, Grohe and Kreutzer~\cite{bib-dawar07+} and
to the very general graph classes with bounded expansion,
which were introduced in~\cite{bib-nes0,bib-nes1,bib-nes2,bib-nes3},
by Dawar and Kreutzer~\cite{bib-dawar09+} (see~\cite{bib-grohe11+} for further details) and,
independently, by Dvo\v{r}\'ak, Kr\'al' and Thomas~\cite{bib-dvorak10+,bib-dvorak13+}.
This line of research ultimately culminated with the result of Grohe, Kreutzer and Siebertz~\cite{bib-grohe14+},
who proved that the first order model checking is fixed-parameter tractable in nowhere-dense classes of graphs
by giving an almost linear time algorithm for this problem when parameterized by the class and the property.

The results that we have just mentioned concern classes of sparse graphs.
While they cannot be extended to all somewhere-dense classes of graphs, see e.g.~\cite{bib-dvorak13+},
it is still possible to hope for proving tractability results for dense graphs that
possess structural properties making first order model checking feasible.
For example, a well-known theorem of Courcelle, Makowsky and Rotics~\cite{bib-courcelle00+}
on monadic second order model checking in classes of graphs with bounded clique-width
implies that first order model checking is tractable for classes of graphs with bounded clique-width;
in relation to the results that we present further,
it is interesting to note that graph classes that can be first order interpreted in classes of graphs with bounded clique-width
also have bounded clique-width~\cite[Corollary 7.38]{bib-courcelle12+}).
Another approach is studying graphs defined by geometric means~\cite{bib-ganian15+,bib-hlineny17+,bib-ken17+}.
The approach that we are interested in here lies in considering graph classes derived
from sparse graph classes by first order interpretations as in~\cite{bib-gajarsky16+,bib-gajarsky18+};
the definition of a first order interpretation can be found in Section~\ref{sec-prelim}.

Specifically,
we are motivated by the following very general folklore conjecture.
\begin{conjecture}
\label{conj-folklore}
The first order model checking is fixed parameter tractable in $I(\GG)$
when parameterized by a graph class $\GG$ with bounded expansion,
a simple first order graph interpretation scheme $I$ and a first order property to be tested.
\end{conjecture}
The first step towards this conjecture was obtained in~\cite{bib-gajarsky16+},
where it was shown that Conjecture~\ref{conj-folklore} holds for classes of graphs with bounded maximum degree.
\begin{theorem}
\label{thm-maxdeg}
The first order model checking is fixed parameter tractable in $I(\GG)$
when parameterized by a class $\GG$ of graphs with bounded maximum degree,
a simple first order graph interpretation scheme $I$ and a first order property to be tested.
\end{theorem}
A combinatorial characterization of classes of graphs interpretable in graph classes of bounded expansion
was given in~\cite{bib-gajarsky18+}.
However,
the characterization does not come with an efficient algorithm to compute the corresponding decomposition.
So, Conjecture~\ref{conj-folklore} remains open.
The approach taken in this paper can be seen as complementary to the one used in~\cite{bib-gajarsky18+}
since we attempt to directly reverse the effect of the first order interpretation.

To motivate our approach, we sketch the proof of Theorem~\ref{thm-maxdeg} from~\cite{bib-gajarsky16+}.
The core of the proof lies in considering first order graph interpretation schemes $I$
where the vertex sets of $G$ and $I(G)$ are the same and
constructing an algorithm that recovers a graph $H$ from $I(G)$ such that
the graphs $G$ and $H$ have the same vertex set and they agree on most of the edges.
We now describe the approach from~\cite{bib-gajarsky16+} phrased in the terminology used in this paper.

We start with introducing additional notation.
A {\em pattern} is a graph $R$ that may contain loops and
it does not contain a pair of adjacent twins that both have loops, or
a pair of non-adjacent twins that neither of them has a loop,
i.e., a graph that has no non-trivial induced endomorphism.
To make our exposition more transparent,
we will further refer to vertices of patterns as to {\em nodes} and generally denote them by $u$ with different subscripts and superscripts;
vertices of graphs that are not patterns will generally be denoted by $v$ with different subscripts and superscripts.
Let $G$ be a graph, $R$ a pattern and
$(V_u)_{u\in V(R)}$ a partition of the vertices of $G$ into parts indexed by the nodes of $R$.
The graph $G^R$ is the graph with the same vertex set as $G$ such that
if $v,v'\in V(G)$, $v\in V_u$ and $v'\in V_{u'}$,
then $vv'$ is an edge in $G^R$ if and only if
either $vv'$ is an edge of $G$ and $uu'$ is not an edge of $R$ or
$vv'$ is not an edge of $G$ or $uu'$ is an edge of $R$.
Alternatively, we may define the graph $G^R$ to be the graph obtained from $G$
by complementing all edges inside sets $V_u$ for each node $u$ with a loop and
between sets $V_u$ and $V_{u'}$ for each edge $uu'$ of $R$.
Note that the graph $G^R$ depends on the chosen partition of the vertex set of $G$;
this partition will always be clear from the context.

A very simple example of the introduced notion is a pattern $R$ that consists of a single node $u$ with a loop.
For every graph $G$, there is only one single class partition of $V(G)$, i.e., $V_u = V(G)$, and
$G^R$ is then the complement of $G$. 
Similarly, if $R$ has two vertices $u$ and $u'$ and the edges $uu$ (loop) and $uu'$, and 
the vertex set of a graph $G$ is partitioned into sets $V_u$ and $V_{u'}$,
then $G^R$ is obtained from $G$ by complementing all edges inside $V_u$ and
all edges between $V_u$ and $V_{u'}$.

We now continue with the exposition of the proof of Theorem~\ref{thm-maxdeg} from~\cite{bib-gajarsky16+}.
Simple first order graph interpretation schemes of graphs with bounded maximum degree
are very closely linked to patterns as given in the next proposition,
which directly follows from Gaifman's theorem~\cite{bib-gajarsky16+}. 
The proposition essentially says that for every integer $d$ and  interpretation scheme $I$,
there exist a pattern $R$ and an integer $D$ such that
the graph $I(G)$ for any graph $G$ with maximum degree $d$
is equal to $H^R$ for a suitable graph $H$ with maximum degree $D$;
note that $R$ and $D$ depend on $I$ and $d$ only.
\begin{proposition}
\label{prop-maxdeg1}
Let $\mathcal{G}_d$ be the class of graphs of maximum degree $d$ and
$I$ a simple first order graph interpretation scheme. 
There exists an integer $D$ and a pattern $R$ such that
for every graph $I(G)$ obtained from $G \in \mathcal{G}_d$ 
there exists a graph $H \in \mathcal{G}_D$ and
a partition $(V_u)_{u\in V(R)}$ of the vertex set of $H$ such that
the graphs $I(G)$ and $H^R$ are the same.
\end{proposition}
This characterization of graphs that can be interpreted in a class of graphs with bounded maximum degree
is then combined with the following ``recovery'' algorithm,
which is implicit in~\cite{bib-gajarsky16+}, to get a proof of Theorem~\ref{thm-maxdeg}.
Note the algorithm \AAA{} from Theorem~\ref{thm-maxdeg2} has two parameters,
one controls the complexity of the structure of a graph and
the other controls the complexity of its transformation.
\begin{theorem}
\label{thm-maxdeg2}
There exists an algorithm \AAA{} that
is fixed parameter with respect to an integer parameter $D$ and a pattern $R$ and
has the following property:
for all $D$ and $R$,
there exist an integer $D'$ and a pattern $R'$ such that
the algorithm \AAA{} takes as an input a graph $G^R$, where $G$ is a graph with maximum degree at most $D$, and
outputs a graph $H$ such that $G^R$ and $H^{R'}$ are the same and the maximum degree of $H$ is at most $D'$.
\end{theorem}
One of our main results is an extension of Theorem~\ref{thm-maxdeg2} to classes of $d$-degenerate graphs.
Note that such graph classes include classes with bounded expansion concerned by Conjecture~\ref{conj-folklore}.
This may look like an innocent extension of Theorem~\ref{thm-maxdeg2} at the first sight.
However, the proof of Theorem~\ref{thm-maxdeg2} relies on the fact that
the degrees of any two vertices of $G^R$ that are contained in the same part $V_u$, $u\in V(R)$, differ by at most $2d$,
i.e., it is easy to recognize vertices that belong to the same part.
This is far from being true in the setting of $d$-degenerate graphs,
which leads to a need for a much finer analysis of the structure of an input graph.
\begin{theorem}
\label{thm-main}
There exists an algorithm \AAA{} that
is fixed parameter with respect to integer parameters $d$ and $K$ and
has the following property:
for all $d$ and $K$,
there exists an integer $m$ such that
the algorithm \AAA{} takes as an input a graph $G^R$ and integers $d$ and $K$,
where $G$ is a $d$-degenerate graph and $R$ is a $K$-node pattern (both unknown to \AAA), and
outputs a graph $H$ such that $G$ and $H$ agree on all but at most $m$ vertices.
In particular, the graph $H$ is $(d+m)$-degenerate.
\end{theorem}

We next present a corollary of Theorem~\ref{thm-main}, which we believe to be of independent interest.
First observe that complementing edges between two subsets $V$ and $V'$ of the vertex set of $G$
is equivalent to complementing on the following three subsets of vertex set: $V\cup V'$, $V$ and $V'$.
Hence, the graph $G^R$ is obtained from $G$ by complementing on at most $K+{K\choose 2}$ subsets of vertices of $G$,
where $K$ is the number of nodes of $R$.
In the other direction, if a graph $H$ is obtained from $G$ by complementing on at most $k$ subsets of vertices,
there exists a pattern $R$ with at most $2^k$ nodes such that $H=G^R$.
Hence, Theorem~\ref{thm-main} implies the following.
\begin{corollary}
\label{cor-main}
There exists an FPT algorithm \AAA{} with the following property:
for every integer $d$ and an integer $k$,
there exists an integer $m$ such that
the algorithm \AAA{} takes as an input a graph $G'$ obtained 
from a $d$-degenerate graph $G$ by complementing on at most $k$ subsets of the vertex set of $G$ and
outputs a graph $H$ such that $G$ and $H$ agree on all but at most $m$ vertices.
\end{corollary}

In relation to the first order model checking,
Corollary~\ref{cor-main} yields the following theorem, which we prove in Section~\ref{sec-FO}.
\begin{theorem}
\label{thm-model}
Let $\GG$ be a graph class with bounded expansion and
let $\GG^k$ be the class containing all graphs that
can obtained from a graph $G\in\GG$ by complementing on at most $k$ subsets of the vertex set of $G$.
For every $k$, the first order model checking is fixed parameter tractable on $\GG^k$.
\end{theorem}
Observe that Proposition~\ref{prop-maxdeg1} implies the following:
if $\GG$ is a class of graphs with bounded maximum degree and
$I$ is a simple first order graph interpretation scheme,
then $I(\GG)\subseteq \GG_D^k$ for some integers $D$ and $k$,
where $\GG_D$ is the class of all graphs with maximum degree at most $D$.
Hence, Theorem~\ref{thm-model} gives an alternative proof of Theorem~\ref{thm-maxdeg}.
On the other hand,
since Proposition~\ref{prop-maxdeg1} does not hold in the setting of graph classes with bounded expansion,
Theorems~\ref{thm-main} and~\ref{thm-model} do not yield an analogous result in this more general setting,
which is concerned by Conjecture~\ref{conj-folklore};
we discuss further details in Section~\ref{sec-concl}.

\section{Preliminaries}
\label{sec-prelim}

In this section, we briefly introduce the notation used throughout the paper, and
present the concepts that we need further.

Graphs considered in this paper are simple, i.e., they do not contain loops or parallel edges
unless stated otherwise.
If $G$ is a graph, then $V(G)$ denotes the set of its vertices.
The {\em neighborhood} of a vertex $v$ in a graph $G$, denoted by $N_G(v)$,
is the set of all vertices adjacent to $v$.
The {\em degree} of a vertex $v$ of a graph $G$ is the size of its neighborhood, and
the {\em relative degree} of $v$ with respect to a subset $X\subseteq V(G)$ is the number of the neighbors of $v$ in $X$.
If $G$ is a graph and $W$ a subset of its vertices,
then the subgraph of $G$ {\em induced} by $W$, denoted by $G[W]$, is the subgraph of $G$
with the vertex set $W$ such that two vertices are adjacent in $G[W]$ if and only if they are adjacent in $G$.
Finally, a graph $G$ is {\em $d$-degenerate}, if its vertices can be ordered in such a way that
each vertex has at most $d$ of its neighbors preceding it.

Let $G$ be a graph. Two vertices $v$ and $v'$ of $G$ are {\em twins}
if every vertex $w$ different from $v$ and $v'$ is adjacent to either both $v$ and $v'$ or none of them.
The binary relation of ``being a twin'' on $V(G)$ is an equivalence relation;
we will call the equivalence classes of this relation {\em twin-classes}.
Note that each twin-class induces either a complete subgraph or an empty subgraph of $G$.

A graph $G'$ is an {\em $r$-shallow minor} of a graph $G$
if it can be obtained from a subgraph of $G$ by contracting vertex-disjoint subgraphs of radii at most $r$ (and removing
arising loops and parallel edges).
We say that a graph class $\GG$ has {\em bounded expansion}
if $\GG$ is monotone, i.e., closed under taking subgraphs, and
there exists a function $f:\NN\to\NN$ such that
the average degree of every $r$-shallow minor of any graph from $G$ is at most $f(r)$.
As we have already mentioned,
examples of classes of graphs with bounded expansion
are classes of graphs with bounded maximum degree and minor-closed classes of graphs.
The latter include classes of graphs with bounded tree-width or graphs embeddable in a fixed surface.

If $G$ is a graph, then a $K$-apex of $G$ is a graph obtained by at adding at most $K$ vertices to $G$ and
joining them to the remaining vertices and between themselves arbitrarily.
The next proposition easily follows from the basic results on classes of graphs with bounded expansion;
see e.g.~\cite[Chapter 5]{bib-nom-book} for further details.

\begin{proposition}
\label{prop-apex}
Let $\GG$ be a class of graph with bounded expansion, and let $K$ be a positive integer.
The class of graphs formed by $K$-apices of graphs from $\GG$ has bounded expansion.
\end{proposition}

Finally,
a \emph{simple first order interpretation scheme} $I$
consists of a pair of formulas $\psi_V(x)$ and $\psi_E(x,y)$.
If $G$ is a graph,
then the graph $I(G)$ has vertex set equal to the set $\{v \in V(G)~|~G \models \psi_V(x)\}$,
i.e., it is the  subset of vertices $x$ of $G$ such that $\psi_V(x)$ holds, and
two vertices $u$ and $v$ of $I(G)$ are adjacent iff $G \models \psi_E(u,v)\vee\psi_E(v,u)$.

\section{Recovering degenerate graphs}

This section is devoted to the proof of Theorem~\ref{thm-main}, one of our two main results.
We need to start with introducing additional notation that will be used in our analysis of complemented graphs.
Let $G$ be a graph.
Two subsets $X$ and $Y$ of the vertex set $V(G)$ are {\em $k$-similar}
if their symmetric difference is at most $k$, i.e., $|X\triangle Y|\le k$.
We say that two vertices of $G$ are {\em $k$-similar} if their neighborhoods are $k$-similar, and
we define the {\em $k$-similarity graph} of $G$ to be the graph with the vertex set $V(G)$
where two vertices are adjacent if they are $k$-similar.

Further fix a pattern $R$ and a partition $(V_u)_{u\in V(R)}$ of $V(G)$.
If $u$ is a node of $R$,
then the {\em $u$-perfect} set is the union of the sets $V_{u'}$ where the union is taken over all neighbors $u'$ of $u$ in $R$.
Note that the $u$-perfect set includes $V_u$ iff $u$ has a loop.
A subset $X$ of the vertex set of $G$ is {\em $(u,k)$-perfect} if it is $k$-similar to the $u$-perfect set, and
a vertex of $G$ is {\em $(u,k)$-perfect} if its neighbors in $G^R$ form a {\em $(u,k)$-perfect} set.
In particular, when saying that a vertex of $G$ is {\em $(u,k)$-perfect},
this always concerns its neighborhood in $G^R$ or in the induced subgraph of $G^R$.

Our goal is to approximately recover graph $G$ from $G^R$
given the size $K$ of $R$ and assuming that $G$ is $d$-degenerate.
We achieve this by finding a partition of $V(G^R)$ that approximates the partition of $(V_u)_{u\in V(R)}$ of $V(G)$.
To find the approximate partition, we use $(u,k)$-perfect vertices introduced above:
if we identify a $(u,k)$-perfect vertex for each class $V_u$ of $(V_u)_{u\in V(R)}$,
then the structure of the neighborhoods of these vertices leads to a good approximation of the partition $(V_u)_{u\in V(R)}$. 
The structural lemmas presented in the next subsection
lead to a simple condition (Lemma~\ref{lm-max-degree})
that allows us to find a $(u,C)$-perfect vertex in the input graph,
where the constant $C$ depends on $d$ and $K$ only.
The presented structural results are then be used to design Algorithm~\ref{alg}, 
which outputs an approximation of the graph $G$. 

\subsection{Structural results}

In this subsection, we present structural results on complemented graphs.
These results will be used in the next subsection to analyze our algorithm.
We start with observing that most vertices of each substantially large part are almost perfect.

\begin{lemma}
\label{lm-in-fract}
Let $R$ be a $K$-node pattern,
$G$ a $d$-degenerate graph with a vertex partition $(V_u)_{u\in V(R)}$, and
$M$ the maximum size of a part $V_u$, $u\in V(R)$.
If a part $V_u$, $u\in V(R)$, contains at least $\frac{M}{4K}$ vertices,
then it contains at least $\left(1-\frac{1}{10K}\right)|V_u|$ vertices that are $(u,80dK^3)$-perfect.
\end{lemma}

\begin{proof}
Fix a node $u$ such that the size of the part $V_u$ is at least $\frac{M}{4K}$, and
observe that a vertex $v$ of $V_u$ is $(u,80dK^3)$-perfect if and only if
its degree in $G$ is at most $80dK^3$.
Hence, we need to show that at least $\left(1-\frac{1}{10K}\right)|V_u|$ vertices of $V_u$
have degree at most $80dK^3$.

Suppose that more than $\frac{1}{10K}|V_u|$ vertices of $V_u$ have degree strictly larger than $80dK^3$.
This implies that the sum of the degrees of the vertices of $V_u$ is strictly larger than 
$$8dK^2|V_u|\ge 2dKM\;\mbox{.}$$
This is impossible since $G$ contains at most $dn\le dKM$ edges in total and
thus the sum of the degrees of all vertices of $G$ is at most $2dKM$.
The statement of the lemma now follows.
\end{proof}

The next lemma shows that almost all vertices with similar neighborhoods must belong to the same part.

\begin{lemma}
\label{lm-out-const}
Let $R$ be a $K$-node pattern and
$G$ a $d$-degenerate graph with a vertex partition $(V_u)_{u\in V(R)}$
such that each $V_u$ contains at least $330dK^3$ vertices.
For every $W\subseteq V(G)$,
there exists a node $u\in V(R)$ such that 
all but at most $330dK^4$ vertices with their neighborhoods $(160dK^3)$-similar to $W$ in $G^R$ belong to $V_u$.
\end{lemma}

\begin{proof}
Suppose that the statement is false and fix a set $W$ that violates the statement.
This implies that there are two different nodes $u$ and $u'$ such that
each of the sets $V_u$ and $V_{u'}$ contains at least $330dK^3$ vertices with $(160dK^3)$-similar to $W$ in $G^R$.
Indeed, take a node $u\in V(R)$ such that $V_u$ contains the largest number of vertices
with their neighborhoods $(160dK^3)$-similar to $W$ in $G^R$;
note that $V_u$ contains at least $330dK^3$ such vertices (otherwise,
any node $u$ would satisfy the statement of the lemma
since there would be at most $330dK^4$ such vertices in total).
Since the set $W$ violates the statement,
there are at least $330dK^4$ vertices with their neighborhoods $(160dK^3)$-similar to $W$ in $G^R$
that do not belong to $V_{u}$.
This implies that there exists a node $u'\in V(R)$ such that $V_{u'}$ also contains at least $330dK^3$ such vertices.

To simplify our notation, fix $n$ to be $330dK^3$.
Choose an $n$-vertex subset $A$ of $V_u$ such that their neighborhoods are $(160dK^3)$-similar to $W$ and
an $n$-vertex subset $A'$ of $V_{u'}$ such that their neighborhoods are $(160dK^3)$-similar to $W$.
Observe that any two vertices in $A\cup A'$ are $(320dK^3)$-similar.

We next distinguish three cases based on whether the nodes $u$ and $u'$ have loops in $R$ and
whether they are adjacent in $R$.
\begin{itemize}
\item {\bf At least one of the two nodes, say $u$, has a loop, and $R$ does not contain the edge $uu'$.}\\
      The subgraph $G[A\cup A']$ contains at most $2dn$ edges,
      which yields that the sum of the degrees of the vertices of $G[A\cup A']$ is at most $4dn$.
      We next compare relative degrees of the vertices of $A\cup A'$ with respect to $A$ in $G^R$.
      Since the neighbors of the vertices of $A'$ in $A$ are the same in $G[A\cup A']$ and in $G^R[A\cup A']$,
      the sum of the relative degrees of the vertices of $A'$ with respect to $A$ is at most $4dn$.
      On the other hand, the sum of the relative degrees of the vertices of $A$ in $G^R[A\cup A']$ is at least $n(n-1)-4dn$.
      Since any two vertices in $A\cup A'$ are $(320dK^3)$-similar in $G^R$ and thus in $G^R[A\cup A']$,
      their relative degrees in $G^R$ with respect to $A$ differ by at most $320dK^3$.
      Consequently, the sums of the relative degrees of the vertices of $A$ and those of $A'$ with respect to $A$ in $G^R$
      can differ by at most $320dK^3n$.
      However, the difference of these two sums is at least
      $$n(n-1)-8dn=n(n-1-8d)\ge n (330dK^3-1-8d)\ge 321dK^3n >320dK^3n\;\mbox{.}$$
\item {\bf At least one of the two nodes, say $u$, does not have a loop, and $R$ contains the edge $uu'$.}\\
      An analogous argument to that used in the first case yields that
      the sum of the relative degrees of the vertices of $A$ with respect to $A$ in $G^R$ is at most $4dn$ and
      the sum of the relative degrees of the vertices of $A'$ with respect to $A$ in $G^R$ is at least $n^2-4dn$.
      Consequently, the difference of these two sums is at least $n^2-8dn>320dK^3n$
      while it cannot exceed $320dK^3n$.
\item {\bf The nodes $u$ and $u'$ either both have loops and are adjacent or
           both do not have a loop and are non-adjacent in $R$.}\\
      Since $R$ is a pattern, there must exist a node $u''$, which is different from $u$ and $u'$, such that
      either $uu''$ is not an edge and $u'u''$ is an edge, or vice versa.
      By symmetry, we can assume the former to be the case.
      Let $A''$ be a set of $n$ vertices contained in $V_{u''}$.
      The number of edges between $A$ and $A''$ in $G$ is at most $2dn$.
      Hence, the sum of the relative degrees of the vertices of $A$ with respect to $A''$ is at most $2dn$
      both in $G$ and in $G^R$.
      On the other hand,
      the sum of the relative degrees of the vertices of $A'$ with respect to $A''$ is at most $2dn$ in $G$, and
      thus at least $n^2-2dn$ in $G^R$.
      Since any two vertices of $A\cup A'$ are $(320dK^3)$-similar,
      their relative degrees with respect to $A''$ in $G^R$ can differ by at most $320dK^3$.
      Consequently, the sums of the relative degrees of the vertices of $A$ and $A'$ can differ by at most $320dK^3n$.
      However, the difference of the two sums is at least $n^2-4dn>320dK^3n$.
\end{itemize}
In each of the three cases, we have obtained a contradiction, which concludes the proof of the lemma.
\end{proof}

To prove the next lemma, we need to introduce some additional notation.
Let $G$ be a graph, $R$ a pattern, $(V_u)_{u\in V(R)}$ a partition of $V(G)$, and $U$ a subset of the nodes of $R$.
The graph $R\setminus U$ need not be a pattern
but there is a unique pattern to that $R\setminus U$ has an induced homomorphism.
This pattern can be obtained as follows.
Let $R'$ be $R\setminus U$.
As long as $R'$ contains either two adjacent twins $u$ and $u'$ that both have loops or
two non-adjacent twins $u$ and $u'$ that none of them has a loop,
identify the nodes $u$ and $u'$ and merge the parts $V_u$ and $V_{u'}$.
The resulting pattern $R_0$ is called the {\em reduction} of $R\setminus U$;
the reduction $R_0$ is uniquely determined by the pattern $R$ and the set $U$.
If $W$ is the union of $V_u$ with $u\not\in U$,
then the new parts $V_u$ indexed by $u\in V(R_0)$ form a partition of the vertex set $G[W]$.
This partition is called the {\em reduced partition} and
it is easy to observe that the graphs $G^R[W]$ and $G[W]^{R_0}$ are the same.

\begin{lemma}
\label{lm-max-degree}
Let $R$ be a $K$-node pattern and
$G$ a $d$-degenerate graph with a vertex partition $(V_u)_{u\in V(R)}$.
If $G$ has at least $1100dK^5$ vertices,
then the vertex of the maximum degree in the $(160dK^3)$-similarity graph of $G^R$
is $(u,570dK^4)$-perfect for some $u\in V(R)$.
\end{lemma}

\begin{proof}
Let $u_M$ be the node of $R$ such that $V_{u_M}$ is the largest part of the  partition $(V_u)_{u\in V(R)}$ and
let $M$ be its size, i.e., $M=|V_{u_M}|$.
Observe that $M\ge 1100dK^4$.
By Lemma~\ref{lm-in-fract}, $V_{u_M}$ contains at least $\left(1-\frac{1}{10K}\right)M\ge \frac{9M}{10}$ vertices that
are $(u_M,80dK^3)$-perfect.
All these vertices are mutually adjacent in the $(160dK^3)$-similarity graph of $G^R$,
which implies that the maximum degree of the $(160dK^3)$-similarity graph of $G^R$ is at least $9M/10-1$.
Let $w$ be the vertex of the maximum degree in the $(160dK^3)$-similarity graph of $G^R$,
$W$ the neighborhood of $w$ in $G^R$, and
$W_s$ the neighborhood of $w$ in the $(160dK^3)$-similarity graph.
Note that $|W_s|\ge 9M/10-1$ and each vertex of $W_s$ is $(160dK^3)$-similar to $w$ in $G^R$.

Let $U'$ be the set of the nodes $u\in V(R)$ such that $|V_u|\le 330dK^3$, and
let $V'$ be the union of the parts $V_u$ with $u\in U'$.
Observe that $|V'|\le 330dK^4$.
Let $R_0$ be the reduction of $R\setminus U'$, and
let $G_0$ be the graph $G\setminus V'$ with the reduced partition $V_{0,u}$, $u\in V(R_0)$.
Observe that $G_0^{R_0}=G^R\setminus V'$ and
each part $V_{0,u}$, $u\in V(R_0)$, has at least $330dK^3$ vertices.
Further, let $W_0=W\setminus V'$, and
note that $W_0$ is the neighborhood of $w$ in $G_0^{R_0}$ and
that each vertex of $W_s\setminus V'$ is $(160dK^3)$-similar to $w$ in $G_0^{R_0}$.

We now apply Lemma~\ref{lm-out-const} to the graph $G_0$ with the pattern $R_0$ and the set $W_0$.
The lemma implies that
there exists a node $u_0$ of $R_0$ such that
there are at most $330dK^4$ vertices outside $V_{0,u_0}$ with their neighborhood $(160dK^3)$-similar to $W_0$ in $G_0^{R_0}$.
Hence, the set $W_s\subseteq V(G)$ contains at most $660dK^4$ vertices that are not contained in $V_{0,u_0}$:
all such vertices are contained in $V'$ or are $(160dK^3)$-similar to $w$ in $G_0^{R_0}$.
It follows that the part $V_{0,u_0}$ contains at least $9M/10-1-660dK^4\ge 9M/10-661dK^4\ge M/4$ vertices of $W_s$.
In particular, the part $V_{0,u_0}$ contains at least $M/4$ vertices in total.

By Lemma~\ref{lm-in-fract},
the part $V_{0,u_0}$ contains at least $\left(1-\frac{1}{10K}\right)|V_{0,u_0}|$ vertices
that are $(u_0,80dK^3)$-perfect with respect to the graph $G_0$ and the pattern $R_0$,
i.e., there are at most $\frac{|V_{0,u_0}|}{10K}\le M/10$ vertices of $V_{0,u_0}$ that
are not $(u_0,80dK^3)$-perfect.
Hence, there is a vertex $v$ that is contained in $W_s\cap V_{0,u_0}$ and
that is $(u_0,80dK^3)$-perfect with respect to the graph $G_0$ and the pattern $R_0$.

Since the vertex $v$ is $(u_0,80dK^3)$-perfect with respect to the graph $G_0$ and the pattern $R_0$,
there exists a node $u\in V(R)$ such that
the vertex $v$ is $(u,80dK^3+|V'|)$-perfect with respect to the graph $G$ and the pattern $R$,
i.e., $v$ is $(u,80dK^3+330dK^4)$-perfect.
Since the vertex $v$ is contained in $W_s$, i.e., it is a neighbor of $w$ in the $(160dK^3)$-similarity graph,
we get that the vertex $w$ is $(u,240dK^3+330dK^4)$-perfect.
Since $240dK^3+330dK^4\le 570dK^4$, the lemma now follows.
\end{proof}

\subsection{Algorithm}\label{ss:alg}

We are now ready to present an algorithm that
can be used to recover the original $d$-degenerate graph $G$ from the graph $G^R$
where $R$ is an a priori unknown $K$-pattern.
The algorithm is given as Algorithm~\ref{alg}.
The algorithm takes the graph $G^R$ as an input and outputs a graph $F$ that
differs from the perfect blow-up $E^R$ of the pattern $R$ only on constantly many vertices,
where $E$ is the graph with the vertex set $V(G)$ and no edges.
Algorithm~\ref{alg} is analyzed in the next lemma.

\begin{algorithm}
{\bf Input:} a graph $H$, integers $d$ and $K$\\
{\bf Output:} a graph $F$ on the vertex set $V(H)$\\
$W:=V(H)$\;
$F:=$empty graph on the vertex set $V(H)$\;
$k:=0$\;
$\SSS:=\emptyset$\;
\While{$|W|\ge 1100dK^5$}{
  \eIf{$\exists\;v\in W$ s.t. $N_{H[W]}(v)$ is $(1140dK^4)$-similar to a set $S_i\cap W$, $S_i\in\SSS$}{
    $W:=W\setminus\{v\}$\;
    join $v$ in $F$ to all the vertices of $S_i\cap W$\;
    }{
    $v:=$max.~degree vertex in the $(160dK^3)$-similarity graph of $H[W]$\;
    $k:=k+1$\;
    $S_k:=$the neighbors of $v$ in $H[W]$\;
    add $S_k$ to $\SSS$\;
    }
  }
output $F$.\\
\caption{Algorithm producing an approximation of the perfect blow-up of an unknown $K$-node pattern.}
\label{alg}
\end{algorithm}

\begin{lemma}
\label{lm-deletion}
Let $R$ be a $K$-node pattern and
$G$ a $d$-degenerate graph with a vertex partition $(V_u)_{u\in V(R)}$.
Suppose that Algorithm~\ref{alg} is applied for $H=G^R$, $d$ and $K$, and
the algorithm outputs a graph $F$.
There exists a subset $U$ of at most $4000dK^6$ vertices of $H$ such that
the graph $F\setminus U$ and $E^R\setminus U$ are the same,
where $E$ is the empty graph with the vertex set $V(H)$.
\end{lemma}

\begin{proof}
Let $W_i$ be the set $W$ at the point when the set $S_i$ is fixed by Algorithm~\ref{alg}, and
let $k$ be the final value of this variables at the end of the algorithm.
Further let $W_0$ be the set $W$ at the end of the algorithm.
By Lemma~\ref{lm-max-degree}, the set $S_i$ is $(u_i,570dK^4)$-perfect in $H[W_i]$ for some $u_i\in V(R)$.
Note that the set $S_i$ is $(u_i,570dK^4)$-perfect in $H[W_j]$ for every $j=i+1,\ldots,k$,
since this property cannot be affected by deleting vertices.
At the point when the set $S_i$ was fixed,
the set $S_i$ was not $(1140dK^4)$-similar to any of the sets $S_1\cap W_i,\ldots,S_{i-1}\cap W_i$.
It follows that the nodes $u_1,\ldots,u_k$ are mutually distinct, which implies $k\le K$.

Let $T_i$ be the set of at most $570dK^4$ vertices of $H[W_i]$ such that
$S_i$ is $u_i$-perfect in $H[W_i\setminus T_i]$, and
let $T=T_1\cup\cdots\cup T_k$.
Further, for each node $u\in V(R)$,
let $V'_u$ be the last $1143dK^4$ vertices of $V_u\setminus T$ removed by Algorithm~\ref{alg} from the set $W$
if such vertices exist; otherwise, let $V'_u=V_u\setminus (T\cup W_0)$.
Note that $|V'_u|\le 1143dK^4$ in either of the cases.

Consider the point when the algorithm removes a vertex $v\in V_{u}$ from the set $W$
because the neighborhood of $v$ is $(1140dK^4)$-similar to the set $S_i\cap W$,
where $W$ is the value of the variable at the time of the removal of $v$.
We say that the vertex $v$ is {\em $u'$-erroneous} for $u'\in V(R)$
if at least one of the vertices of $V_{u'}\setminus (V'_{u'}\cup T\cup W_0)$ has not yet been removed from $W$ and
\begin{itemize}
\item either $uu'$ is an edge of $R$ but $V'_{u'}$ and $S_i$ are disjoint, or
\item $uu'$ is not an edge of $R$ but $V'_{u'}$ is a subset of $S_i$.
\end{itemize}
Note that it can be the case that the nodes $u$ and $u'$ in the above definition coincide, and
a vertex $v$ can be $u'$-erroneous for several choices of $u'$.
Also note that if $v$ is $u'$-erroneous, then $V_{u'}\setminus (V'_{u'}\cup T\cup W_0)\not=\emptyset$,
which implies that $|V'_{u'}|=1143dK^4$.
Let $V_{u,u'}$ be the set of vertices of $V_u\setminus (V'_{u}\cup T)$ that are $u'$-erroneous.

The set $U$ will contain the following vertices:
\begin{itemize}
\item at most $1100dK^5$ vertices contained in $W_0$,
\item at most $k\cdot 570dK^4\le 570dK^5$ vertices contained in $T$,
\item at most $K\cdot 1143dK^4\le 1143dK^5$ vertices contained in the set $V'_u$, $u\in V(R)$, and
\item the vertices of all sets $V_{u,u'}$, $u,u'\in V(R)$.
\end{itemize}
We next show that each of the sets $V_{u,u'}$ contains at most $1143dK^4$ vertices,
which would imply that the size of $U$ does not exceed $4000dK^6$.

Set $n=1143dK^4$ to simplify the notation, and
suppose that there exists a set $V_{u,u'}$ containing more than $n$ vertices for some $u,u'\in V(R)$ (possibly $u=u'$).
Let $X$ be a subset of $V_{u,u'}$ containing exactly $n$ vertices.
Note that that if $u=u'$, the sets $X$ and $V'_{u'}$ are disjoint
because all vertices of $V'_{u'}$ are removed from $W$ after those of $X$.
We first consider the case that $uu'$ is not an edge of $R$,
which includes the case that $u=u'$ and $u$ does not have a loop.
Since the vertices of $X$ are $u'$-erroneous,
the set $V'_{u'}$ contains $n$ vertices and
all vertices of $V'_{u'}$ are removed from $W$ later than the vertices of $X$,
the $d$-degeneracy of $G$ implies that the number of edges between $X$ and $V'_{u'}$ in $G$ is at most $2dn$.
When a vertex $v\in X$ is removed from $W$ by Algorithm~\ref{alg},
it is adjacent to at least $|V'_{u'}|-1140dK^4\ge 3dK^4$ vertices of $V'_{u'}$ in $H=G^R$
since the neighborhood of $v$ is $(1140dK^4)$-similar to $S_i$ and $V'_{u'}\subseteq S_i$.
Hence, the number of edges between $X$ and $V'_{u'}$ in $H=G^R$ is at least $3dK^4n\ge 3dn$.
However, the edges between the vertices of $X$ and those of $V'_{u'}$ are the same in $G$ and $G^R$,
which is impossible.

The other case that we need to consider is that when $uu'$ is an edge of $R$;
this case also includes the case that $u=u'$ and $u$ has a loop.
The arguments are analogous to the first case but we include them for completeness.
We again observe that the number of edges between $X$ and $V'_{u'}$ in $G$ is at most $2dn$.
When a vertex $v\in X$ is removed from $W$,
it is adjacent to at most $1140dK^4$ vertices of $V'_{u'}$ in $H=G^R$
since its neighborhood is $(1140dK^4)$-similar to $S_i$ and the sets $S_i$ and $V'_{u'}$ are disjoint.
It follows that each vertex $v\in X$ is adjacent to at least $|V'_{u'}|-1140dK^4\ge 3dK^4$ vertices of $V'_{u'}$ in $G$.
This implies that the number of edges between $X$ and $V'_{u'}$ in $G$ is at least $3dK^4n\ge 3dn$,
which is again impossible.

To complete the proof of the lemma, we need to show that
the graphs $F\setminus U$ and $E^R\setminus U$ are the same.
Let $v$ and $v'$ be two vertices of $V(H)\setminus U$ such that $v\in V_u$ and $v'\in V_{u'}$.
By symmetry, we can assume that $v$ is removed before $v'$.
Suppose that the vertex $v$ was removed by Algorithm~\ref{alg}
because the neighborhood of $v$ in $H[W]$ was $(1140dK^4)$-similar to a set $S_i$
where $W$ is the value of the set at the time of the removal of $v$ from $W$.
Since the vertex $v'$ does not belong to $U$, it is not contained in $V'_{u'}\cup W_0\cup T$,
which implies that $V'_{u'}\subseteq (V_{u'}\cap W)\setminus T$.
Further, since the set $S_i$ is $u_i$-perfect in $H[W\setminus T]$,
the set $S_i$ either contains $(V_{u'}\cap W)\setminus T$ or is disjoint from $(V_{u'}\cap W)\setminus T$.
Since $v$ is not $u'$-erroneous, the former happens if and only if $uu'$ is an edge in $R$, and
the latter happens otherwise.
Hence, the vertices $v$ and $v'$ are joined by an edge in $F$ if and only if $uu'$ is an edge of $R$.
\end{proof}

Lemma~\ref{lm-deletion} yields the proof of Theorem~\ref{thm-main} as follows.

\begin{proof}[Proof of Theorem~\ref{thm-main}.]
Fix $d$ and $K$, and set $m=4000dK^6$.
Let $G_0$ be the input graph, and
suppose that $G$ is the $d$-degenerate graph and $R$ is the $K$-node pattern such that $G_0=G^R$.
Note that both $G$ and $R$ are not given to the algorithm \AAA.

The algorithm \AAA{} applies Algorithm~\ref{alg} to the graph $G_0$ and integers $d$ and $K$, and
Algorithm~\ref{alg} outputs a graph $F$.
By Lemma~\ref{lm-deletion}, the graphs $E^R$ and $F$ agree on all but at most $m$ vertices,
where $E$ is the empty graph on the same vertex set as $G_0$.
The algorithm \AAA{} then outputs the graph $G_0\triangle F$,
i.e., the graph with the same vertex set as $G_0$ and
with the edge set that is the symmetric difference of the edge sets of $G_0$ and $F$.
Observe that the graph $G=G^R\triangle E^R$ and the output graph $G_0\triangle F=G^R\triangle F$
differ exactly where the graphs $E^R$ and $F$ differ.
It follows that the output graph $G_0\triangle F$ and the graph $G$ agree on all but at most $m$ vertices,
which implies that the output graph $G_0\triangle F$ is $(d+m)$-degenerate.
\end{proof}

\section{FO model checking}
\label{sec-FO}

In this section, we prove Theorem~\ref{thm-model}, which is our second main result, and
also discuss first order model checking in graphs obtained by complementing parts of degenerate graphs.
We start with proving Theorem~\ref{thm-model}.

\begin{proof}[Proof of Theorem~\ref{thm-model}.]
Fix a graph class $\GG$ with bounded expansion and an integer $k$, and set $K=2^k$.
Since the graph class $\GG$ has bounded expansion, there exists an integer $d$ such that
every graph in $\GG$ is $d$-degenerate.
Set $m=4000dK^6$ and
let $\HH$ be the graph class that contain all $m$-apices of subgraphs of graphs contained in $\GG$.
By Proposition~\ref{prop-apex}, the graph class $\HH$ has bounded expansion.

Let $G'$ be a graph obtained from a graph $G\in\GG$ by complementing on at most $k$ subsets of the vertex set of $G$, and
let $V$ be the common vertex set of $G$ and $G'$.
Note that there exists a $K$-node pattern $R$ (which can be chosen independently of $G$ and $G'$
but this fact is not needed in our proof) and
a partition $(V_u)_{u\in V(R)}$ of the vertex set $V$ such that $G'=G^R$.
Apply Algorithm~\ref{alg} to $G'$, $d$ and $K$, and let $F$ be the output graph.
Since the graphs $F$ and $E^R$,
where $E$ is the empty graph on the vertex set $V$,
coincide on all but at most $m$ vertices by Lemma~\ref{lm-deletion},
there exists a $(K+m)$-node pattern $R_F$ such that $F=E^{R_F}$
for a suitable partition $(V'_u)_{u\in V(R_F)}$ of the vertex set $V$.
Moreover, the pattern $R_F$ and the partition $(V'_u)_{u\in V(R_F)}$ can be efficiently constructed:
the at most $K+m$ twin-classes of the graph $F$ form the partition $(V'_u)_{u\in V(R_F)}$ and
the partition into twin-classes uniquely determine the pattern.

Let $H$ be the graph with the vertex set $V$ and
the edge set being the symmetric difference of the edge sets of $G'$ and $F$.
Observe that $H^{R_F}=G'$.
By Lemma~\ref{lm-deletion}, the graphs $G$ and $H$ agree on all but at most $m$ vertices,
which implies that the graph $H$ belongs to the class $\HH$.
The application of the pattern $R_F$ to $H$ can be simulated
by viewing the partition $(V'_u)_{u\in V(R_F)}$ as a vertex $(K+m)$-coloring and
encoding the application of the pattern $R_F$ by a first order formula.
In particular,
there exists a simple first order graph interpretation scheme $I$ of $(K+m)$-vertex colored graphs such that $I(H)=G'$.
Since there are only finitely many choices of $R_F$ (because the number of nodes of $R_F$ is bounded) and
it is possible to use disjoint sets of colors to encode applications of different patterns $R_F$,
there exists such an interpretation scheme $I$ that is universal for all patterns $R_F$.
The fixed parameter tractability of the first order model checking in $\GG^k$ is now implied
by the fixed parameter tractability of the first order model checking in graph classes with bounded expansion that
contain graphs vertex-colored by a bounded number of colors,
which directly follows from the results of~\cite{bib-dawar09+,bib-dvorak10+,bib-dvorak13+}.
\end{proof}

The first order model checking in $d$-degenerate graphs is hard from the point of fixed parameter tractability,
however, many parameterized problems that are hard for general graphs
become fixed parameter tractable when restricted to $d$-degenerate graphs.
Two prominent examples of such problems are the $k$-clique problem,
which asks whether the input graph contains a complete subgraph with $k$ vertices, and
the $k$-independent set problem,
which asks whether the input graph contains $k$ independent vertices.
Both these problems are fixed parameter tractable when parameterized by $d$ and $k$.

To explore hopes of
extending the fixed parameter tractability results for $d$-degenerate graphs
to classes of graphs obtained by complementing $d$-degenerate graphs,
we provide a brief analysis of the fixed parameter tractability of the $k$-clique problem
in graphs obtained from $d$-degenerate graphs by applying patterns
In the rest of this section, $\GG_d$ denotes the class of $d$-degenerate graphs and
$\GG_d^R$ for a pattern $R$ will be the class of all graphs that can be obtained from a graph $G\in\GG_d$ by applying the pattern $R$,
i.e., the class of all graphs $G^R$ for $G\in\GG_d$.
We start with considering the parameterization by both $R$ and $k$,
where the problem turns out to be tractable for $d=1$ and hard for $d\ge 2$
as given in the following two propositions.

\begin{proposition}
\label{prop-deg-cw}
The $k$-clique problem in the class $\GG_1^R$
is fixed parameter tractable
when parameterized by a pattern $R$ and an integer $k$.
\end{proposition}

\begin{proof}
The class $\GG_1$ of $1$-degenerate graphs is the class of all forests.
Recall that a rank-width of a graph $G$ is defined as the minimum $r$ such that
there exists a tree $T$ with leaves one-to-one corresponding to the vertices of $G$ such that
each edge $e$ of $T$ determines a vertex cut $(A,B)$ of $G$ ($A$ and $B$ are the vertices assigned
to the leaves of the two components of $T\setminus e$) such that the adjacency matrix of the cut $(A,B)$ has rank at most $r$.
It is not hard to see that each forest has rank-width at most one.
Next observe that if the adjacency matrix of a vertex cut $(A,B)$ in a graph $G$ has rank $r$,
then the adjacency matrix of the cut $(A,B)$ in $G^R$ has rank at most $r+K$.
Consequently, if $G$ is a graph with rank-width $r$ and $R$ is a $K$-node pattern,
then the rank-width of $G^R$ is at most $r+K$.
We conclude that all graphs contained in the class $\GG_1^R$ have bounded rank-width,
which implies that all graphs contained in the class $\GG_1^R$ have bounded clique-width~\cite{bib-oum06+}.
Since monadic second order model checking is fixed parameter tractable
in classes of graphs with bounded clique-width~\cite{bib-courcelle00+},
the statement of the proposition follows.
\end{proof}

\begin{proposition}
\label{prop-deg-neg}
The $k$-clique problem in the class $\GG_2^R$ is $W[1]$-hard
when parameterized by a pattern $R$ and an integer $k$.
\end{proposition}

\begin{proof}
We present a reduction from the multicolored $k$-clique problem, which is a well-known $W[1]$-hard problem.
The multicolored $k$-clique problem asks whether a given $k$-partite graph contains a clique of order $k$.
Let $G$ be an arbitrary $k$-partite graph,
let $V_1,\ldots,V_k$ be its vertex parts, and
let $H$ be the graph obtained from $G$ by subdividing each edge.
Note that $H$ can be viewed as a $\left(k+{k\choose 2}\right)$-partite graph
with parts $V_1,\ldots,V_k$ and parts $V_{ij}$, $1\le i<j\le k$, formed by vertices of degree two
associated with edges between the parts $V_i$ and $V_j$ in the graph $G$.
Let $R$ be a pattern with $k+{k\choose 2}$ nodes $u_i$, $1\le i\le k$, and $u_{ij}$, $1\le i<j\le k$, such that
$R$ has no loops
but all pairs of nodes of $R$ are joined edges except for pairs $u_i$ and $u_{ij}$ and pairs $u_j$ and $u_{ij}$, $1\le i<j\le k$.
Set $V_{u_i}=V_i$, $1\le i\le k$, and $V_{u_{ij}}=V_{ij}$, $1\le i<j\le k$;
this yields a vertex partition $(V_u)_{u\in V(R)}$ of the graph $H$.
The graph $H^R$ is a $\left(k+{k\choose 2}\right)$-partite graph.
Note that if $k\ge 4$, then $H^R$ contains a clique with $k+{k\choose 2}$ vertices if and only if
$H$ contains a subdivision of a clique with $k$ vertices.
Consequently, if $k\ge 4$, then $G$ contains a clique with $k$ vertices if and only if
$H^R$ contains a clique with $k+{k\choose 2}$ vertices.
Since $H$ is a $2$-degenerate graph, the proposition now follows.
\end{proof}

Proposition~\ref{prop-deg-neg} leaves it open
whether the $k$-clique problem is fixed parameter tractable
when $d$ and $R$ are fixed and $k$ is the parameter.
We address this affirmatively in the next proposition.

\begin{proposition}
\label{prop-deg-pos}
For every integer $d$ and every pattern $R$,
the $k$-clique problem in the class $\GG_d^R$
is fixed parameter tractable
when parameterized by $k$.
\end{proposition}

\begin{proof}
We present an algorithm that decides whether a graph $H\in\GG_d^R$ contains a complete subgraph with $k$ vertices.
In view of Theorem~\ref{thm-main} and Lemma~\ref{lm-deletion},
we may assume (at the expense of considering a larger integer $d$ and a larger pattern $R$) that
the algorithm is given a graph $G\in\GG_d$, a pattern $R$ and a vertex partition $(V_u)_{u\in V(R)}$ such that $H=G^R$.
If $R$ contains a node $u$ with a loop such that $|V_u|>dk$,
then $H$ contains a complete subgraph with $k$ vertices:
indeed, since the subgraph $G[V_u]$ is $(d+1)$-colorable,
$G[V_u]$ contains an independent set of at least $k$ vertices;
this set forms a complete subgraph in $H=G^R$.
Hence, we may assume that the following holds for every node $u$ of $R$: $u$ has no loop or $|V_u|\le dk$.

We next observe that $H[V_u]$ contains at most $\max\{2^{dk},2^d|V_u|\}$ (not necessarily inclusionwise maximal) complete subgraphs.
Indeed, if $|V_u|\le dk$, then there are at most $2^{dk}$ subsets of $V_u$ and the claim follows.
Otherwise, $u$ has no loop and $G[V_u]=H[V_u]$ and the claim follows since $H[V_u]$ is $d$-degenerate.
Let $\CC_u$ be the set of all complete subgraphs of $H[V_u]$ (including the one with no vertices,
i.e., the one induced by the empty set).
The algorithm now tests all possible combinations of subgraphs from $\CC_u$, $u\in V(R)$,
whether they form a complete subgraph in $H$.
This identifies all complete subgraphs of $H$.
The running time of the algorithm is bounded by the product of the sizes of the set $\CC_u$, $u\in V(R)$,
i.e., the algorithm runs in time $O\left(2^{dkK}n^{K+O(1)}\right)$,
where $n$ is the number of vertices of the input graph $H$ and $K$ is the number of nodes of the pattern $R$.
\end{proof}

\section{Conclusion}
\label{sec-concl}

Our results have been motivated by the characterization of graphs that are first interpretable in graphs with bounded maximum degree
as given in Proposition~\ref{prop-maxdeg1}.
While we were able to translate Theorem~\ref{thm-maxdeg2} to the setting of Conjecture~\ref{conj-folklore} and
even the more general setting of degenerate graphs,
Proposition~\ref{prop-maxdeg1} fails to extend to the setting of Conjecture~\ref{conj-folklore},
which we now outline.
Consider a class $\GG$ of all star forests, one of the simplest classes of sparse graphs with unbounded maximum degree, and
also consider the simple first order graph interpretation scheme $I$ such that
two vertices in $I(G)$ are joined by an edge iff their distance in a graph $G$ is at most two.
The graph class $I(\GG)$ contains all graphs $G$ such that each component of $G$ is a complete graph.
Let $\HH$ be a graph class and $R$ a pattern such that $I(\GG)\subseteq\HH^R$,
where $\HH^R$ is the class of graphs $H^R$, $H\in\HH$.
Let $K$ be the number of nodes of $R$ and consider a graph $G\in\GG$ formed by $k\cdot K$ stars each with $k\cdot K-1$ leaves
for an integer $k\ge K+1$.
The graph $I(G)$ consists of $k\cdot K$ cliques each having $k\cdot K$ vertices;
let $C_1,\ldots,C_{k\cdot K}$ be the vertex sets of the $k$ cliques forming the graph $I(G)$.
Suppose that $I(G)=H^R$ for a graph $H\in\HH$ and a vertex partition $(V_u)_{u\in V(R)}$ of $H$.
There exist a node $u$ such that $|V_u\cap C_i|\ge k$ for at least two different indices $i$;
by symmetry we can assume that $|V_u\cap C_1|\ge k$ and $|V_u\cap C_2|\ge k$.
If the node $u$ has a loop in $R$, then the graph $H$ contains all edges between $V_u\cap C_1$ and $V_u\cap C_2$,
i.e., $H$ contains a complete bipartite subgraph with parts of sizes $k$.
If the node $u$ does not have a loop in $R$, then $H[V_u\cap C_1]$ is a complete subgraph with $k$ vertices,
i.e., $H$ contains a complete bipartite subgraph with parts of sizes $\lfloor k/2\rfloor$.
We conclude that the graph class $\HH$ contains graphs with arbitrary large complete bipartite subgraphs;
this implies that the graph class $\HH$ does not have bounded expansion.

In view of the results presented in Section~\ref{sec-FO},
it is natural to wonder about the fixed parameter tractability of other important graph problems.
One of such problems is the $k$-dominating set problem,
which asks whether the input graph contains $k$ vertices such that
each vertex of the graph is one of these $k$ vertices or adjacent to at least one of them.
The $k$-dominating set problem is known to be fixed parameter tractable for $d$-degenerate graphs~\cite{bib-alon07+}
when parameterized by $d$ and $k$.
However, we were not able to resolve the fixed parameter complexity of the $k$-dominating set problem
in graphs obtained by complementing vertex subsets of $d$-degenerate graphs and
even the following particular case seems to be challenging.

\begin{problem}
Is the $k$-dominating set problem in the complements of $d$-degenerate graphs
fixed parameter tractable when parameterized by $d$ and $k$?
\end{problem}

\end{document}